\documentclass[a4paper,11pt]{article}

\usepackage{fullpage}
\usepackage{times}
\usepackage{soul}
\usepackage{url}
\usepackage[utf8]{inputenc}
\usepackage[small]{caption}
\usepackage{graphicx}
\usepackage{xcolor}
\usepackage{amsmath}
\usepackage{booktabs}
\usepackage{algorithm}
\usepackage{algorithmic}
\usepackage{enumitem}

\usepackage{amsthm}
\usepackage{amssymb}

\usepackage{bbm}

\newtheorem{theorem}{Theorem}[section]

\newtheorem{lemma}[theorem]{Lemma}

\theoremstyle{definition}

\usepackage{natbib}

\newcount\Comments  
\newcommand{\kibitz}[2]{\ifnum\Comments=1{\color{#1}{#2}}\fi}

\newcommand{\vv}{\mathbf{v}}

\newcommand{\eq}{\text{EQ}}
\newcommand{\SW}{\text{SW}}
\newcommand{\DI}{\text{DI}}

\newcommand{\PoA}{\text{PoA}}
\newcommand{\PoS}{\text{PoS}}

\newcommand{\calG}{\mathcal{G}}

\allowdisplaybreaks

\title{\bf Swap Stability in Schelling Games on Graphs\thanks{This work has been supported by the European Research Council (ERC) under grant number 639945 (ACCORD), and by the EPSRC International Doctoral Scholars Grant EP/N509711/1.}}

\author{Aishwarya Agarwal, Edith Elkind, Jiarui Gan, and Alexandros A. Voudouris}

\date{Department of Computer Science, University of Oxford}

\begin{document}

\maketitle

\begin{abstract}
We study a recently introduced class of strategic games 
that is motivated by and generalizes Schelling's well-known residential segregation model.
These games are played on undirected graphs, with the set of agents 
partitioned into multiple types; each agent either occupies a node of the graph
and never moves away or aims to maximize the fraction of her neighbors
who are of her own type. We consider a variant of this model that we call swap Schelling games,
where the number of agents is equal to the number of nodes of the graph, and  
agents may {\em swap} positions with other agents to increase their utility.
We study the existence, computational complexity and quality of 
equilibrium assignments in these games, both from a social welfare perspective 
and from a diversity perspective.
\end{abstract}

\section{Introduction}
Segregation is observed in many communities; people tend to group together on the basis 
of politics, religion, or socioeconomic status. 
This phenomenon has been extensively documented in residential 
metropolitan areas, where people may select where to live based on the 
racial composition of the neighborhoods. To formalize and study how the motives 
of individuals lead to residential segregation, \citet{S69,S71} 
proposed the following simple, yet elegant model. There are two types of agents who are to be 
placed on a line or a grid. An agent is happy with her location if at least a fraction 
$\tau \in (0,1]$ of the agents within a certain radius are of the same type as her. 
Happy agents do not want to move, but unhappy agents are willing to do so in hopes 
of improving their current situation. Schelling described a dynamic process 
where at each step unhappy agents jump to random open positions or swap positions 
with other randomly selected agents, and showed that it can lead to a 
completely segregated placement, even if the agents themselves are tolerant of
mixed neighborhoods ($\tau < 1/2$).

Over the years, Schelling's work became very popular among researchers in Sociology and 
Economics, who proposed and studied numerous variants of his model, mainly via agent-based 
simulations; see the paper of \citet{CF08} and references therein for examples of this 
approach. Variants of the model have also been rigorously analyzed in a series of 
papers \citep{Y01,Z04a,BIKK12,BEL14,Bhakta2014clustering,BEL15,IKLZ17}, which showed that the 
random behavior of the agents leads with high probability to the formation of large 
monochromatic regions.

While all these papers focused on settings where the agents' behavior is 
random, it is more realistic to assume instead that the agents are {\em strategic} 
and move only when they have an opportunity to improve their situation. So far,  
only a few papers have followed such a game-theoretic approach.  
In particular, \citet{Z04b} considered a game where the agents optimize a 
single-peaked utility function, and very recently, \citet{CLM18}, \citet{Elkind2019jump} and 
\citet{Echzell2019dynamics} studied strategic settings that are closer to the original model 
of Schelling, but capture more that two agent types and richer graph topologies.

In particular, \citet{CLM18} study a setting with two types of agents, who have preferred 
locations, and can either swap with other agents or jump to empty positions. For a given 
tolerance threshold $\tau \in (0,1]$, each agent's primary goal is to maximize 
the fraction of her neighbors that are of her own type as long as this fraction is below $\tau$
(with all fractions above $\tau$ being equally good); 
her secondary goal is to be as close as possible to her preferred location. 
For both types of games (swap and jump), Chauhan {\em et al.}~identify values of $\tau$ 
for which the best response dynamics of the agents leads to an equilibrium 
when the topology is a ring or a regular graph.
\citet{Echzell2019dynamics} strengthen these results and extend them to more than two
agent types, as well as study the complexity of computing assignments that maximize 
the number of happy agents. 

\citet{Elkind2019jump} consider a similar model with $k$ types; 
however, they treat agents' location preferences differently from Chauhan {\em et al.}
Namely, in their model each agent is either stubborn 
(i.e., has a preferred location and is unwilling 
to move) or strategic (i.e., aims to maximize the fraction of her neighbors 
that are of her own type; this corresponds to setting $\tau=1$ in the model of Chauhan {\em et al.}). 
They focus on jump games, i.e., games where agents may jump to empty positions, 
and analyze the existence and complexity of computing Nash equilibria, as well as prove bounds on
the price of anarchy~\cite{KP99} and the price of stability~\cite{ADKTWR08}. 

\subsection{Our Contribution}
We combine the two approaches by considering swap games in the model of \citet{Elkind2019jump}. 
That is, we assume that the number of agents is equal to the number of nodes in the 
topology, and two agents can swap locations if each of them prefers the other agent's
location to her own.
We begin by studying the existence of equilibrium assignments. 
While such assignments exist for highly structured topologies, 
we show that they may fail to exist in general, even for simple topologies such as trees. 
Moreover, we show that deciding whether an equilibrium exists is NP-complete.  
We also study the quality of assignments in terms of their social welfare: 
we prove bounds on the price of anarchy and the price of stability for many interesting cases, 
and show that computing an assignment with high social welfare is NP-complete; the latter
result complements the result of Elkind {\em et al.}~in that it applies to the case where the number
of agents equals the number of nodes in the topology. 

Given that the goal of Schelling's work was to study integration, 
it is natural to ask what level of integration can be achieved at equilibrium.
There is a number of integration indices that have been proposed 
for this purpose (see, e.g., the survey of \citet{MD88}). 
However, many of the indices defined in the literature are formulated
for settings where the topology is highly regular and there are only two agent types, 
and it is not immediately clear how to adapt them to our model.
We therefore focus on a simple index, which we call the {\em degree of integration},
that is inspired by the work of \citet{LC82} and admits a natural interpretation in our context. 
This index counts the number of agents who are exposed to agents of other types,  
i.e., have at least one neighbor of a different type. 
We then study the price of anarchy and the price of stability with respect
to this index: that is, we compare the value
of our index in the best and worst equilibrium of our game to the optimal value
of this index that can be achieved for a given instance. We note that, to the best
of our knowledge, this is the first result of this type in the context of Schelling games:
the previous work on integration in the Schelling model typically
focused on evaluating a given integration index after some number of steps
of the underlying dynamical process, and did not ask what level of integration
can be achieved if the agents were non-strategic.
We obtain strong negative results:
it turns out that even the best equilibria are often much less diverse
than the maximally diverse assignments. however, when the topology is a line, 
the price of stability with respect to our index can be bounded by a small constant.
We also show that maximizing diversity is computationally hard.

\subsection{Further Related Work}
As mentioned above, Schelling's model has been studied extensively both empirically and 
theoretically. For an 
introduction to the model and a survey of 
its many variants, we refer the 
reader to the book of \citet{EK10}, and the papers 
by \citet{BIKK12} and \citet{IKLZ17}.
Besides the closely related papers by \citet{CLM18}, \citet{Elkind2019jump} and 
\citet{Echzell2019dynamics}, another work that is similar in spirit 
is a recent paper by \citet{Massand2019graphical}, who study swap stability 
in games where a set of items is to be allocated among agents 
who are connected via a social network, so that each agent gets one item, 
and her utility depends on the items she and her neighbors 
in the network get; however, their results are not directly applicable to our setting.
%
Also, Schelling games share a number of properties with 
hedonic games~\cite{DG80,BJ02}, and in particular, with fractional hedonic games 
\cite{ABBHOP19} and hedonic diversity games \cite{BEI19}. However, 
a fundamental difference between hedonic games and Schelling games
is that in the former agents form pairwise disjoint coalitions, while
in the latter the neighborhoods of different nodes of the topology may overlap.

\section{Preliminaries}\label{sec:prelim}
A {\em $k$-swap game} is given by a set $N$ of $n \geq 2$ {\em agents}
partitioned into $k \geq 2$ pairwise disjoint {\em types} $T_1, \ldots, T_k$,
and an undirected simple connected graph $G=(V,E)$ with $|V|=n$, called the {\em topology}.
We often identify types with colors: e.g., in a $2$-swap game 
each agent is either red ($T_1$) or blue ($T_2$).
The agents are also classified as either {\em strategic} or {\em stubborn}.
We denote by $R$ the set of strategic agents and by $S$ the set of stubborn agents, 
so that $R \cup S = N$.
Stubborn agents never move away from the nodes they occupy, while
a strategic agent aims to maximize her personal utility,
and is willing to {\em swap} positions with other agents to achieve this.

Given an agent $i\in T_\ell$, we refer to all other agents in $T_\ell$ as {\em friends
of $i$} and denote the set of $i$'s friends by $F_i = T_\ell \setminus \{i\}$.
Each agent $i$ occupies some node $v_i \in V$ of the topology $G$
so that $v_i \neq v_j$ for every pair of agents $i \neq j$.
The vector $\vv = (v_1,\dots, v_n)$ that lists the locations of all agents 
is called an {\em assignment}. Given an assignment $\vv$, 
we denote by $\pi_v(\vv)$ the agent that occupies node $v \in V$, that is, $\pi_{v_i}(\vv)=i$.

Given an assignment $\vv$,
let $N_i(\vv) = \{ j \neq i: \{v_i,v_j\} \in E \}$ be the set of neighbors of agent $i$.
The utility $u_i$ of a stubborn agent $i\in S$ is independent of the assignment; e.g., we can 
set $u_i(\vv)=0$ for each $i\in S$.
The utility of a strategic agent $i\in R$ for assignment $\vv$ is
\begin{align*}
u_i(\vv) = \frac{|N_i(\vv) \cap F_i|}{|N_i(\vv)|}.
\end{align*}
Observe that, since $|V|=n$, every node is occupied by some agent,
and therefore $N_i(\vv) \neq \varnothing$ for every $i\in N$.

For every assignment $\vv$, let $\vv^{i \leftrightarrow j}$ 
be the assignment that is obtained from $\vv$ by swapping
the positions of agents $i$ and $j$:
$v_\ell^{i \leftrightarrow j} =v_\ell$ for every $\ell\in N\setminus\{i,j\}$,
$v_i^{i \leftrightarrow j} = v_j$ and $v_j^{i \leftrightarrow j} = v_i$.
Agents $i$ and $j$ swap positions if and only if 
they both strictly increase their utility:
$u_i(\vv^{i \leftrightarrow j}) > u_i(\vv)$
and
$u_j(\vv^{i \leftrightarrow j}) > u_j(\vv)$.
Clearly, agents of the same type cannot both increase 
their utilities by swapping, 
so swaps always involve agents of different types. 
An assignment $\vv$ is an {\em equilibrium} if no pair of agents $i,j$ want to swap positions.
That is, $\vv$ is an equilibrium if and only if for every $i, j\in R$
we have $u_i(\vv) \geq u_i(\vv^{i \leftrightarrow j})$
or
$u_j(\vv) \geq u_j(\vv^{i \leftrightarrow j})$.
We denote the set of all equilibrium assignments 
of the $k$-swap game $\calG$ by $\eq(\calG)$.

For every assignment $\vv$, we define two benchmarks that aim to capture,
respectively, the agents' happiness and the societal diversity.
The first one is the well-known {\em social welfare},
defined as the total utility of all strategic agents:
$$ 
\SW(\vv) = \sum_{i \in R} u_i(\vv). 
$$
Our second benchmark is the {\em degree of integration}:
we say that an agent is {\em exposed} is she has at least one neighbor
of a different type, and count the number of exposed agents:
$$ 
\DI(\vv) = |\{i\in R: N_i(\vv) \setminus F_i \neq \varnothing\}|.
$$
Note that we ignore the stubborn agents in the definitions of
our benchmarks, as their utility is always the same 
and they never want to move somewhere else.

For $f \in \{\SW,\DI\}$, let $\vv^*_f(\calG)$ be the {\em optimal} assignment 
in terms of the benchmark $f$ for a given game $\calG$. 
The {\em price of anarchy} (PoA) in terms of the benchmark $f$ 
is the worst-case ratio (over all $k$-swap games $\calG$ such that 
$\eq(\calG)\neq \varnothing$) between the optimal performance 
(among all assignments) and the performance of the 
{\em worst} equilibrium assignment. 
Similarly, the {\em price of stability} (PoS) in terms of $f$ 
is the worst-case ratio between the optimal performance 
and the performance of the {\em best} equilibrium:
\begin{align*}
\PoA_f &= \sup_{\calG: \eq(\calG) \neq \varnothing} 
\sup_{\vv \in \eq(\calG)} \frac{f(\vv^*_f(\calG))}{f(\vv)}, \\
\PoS_f &= \sup_{\calG: \eq(\calG) \neq \varnothing} 
\inf_{\vv \in \eq(\calG)} \frac{f(\vv^*_f(\calG))}{f(\vv)}.
\end{align*}
For readability, we refer to the quantity $\PoA_\SW$ as the {\em social price of 
anarchy} and to $\PoA_\DI$ as the {\em integration price of anarchy}, and use 
similar language for the price of stability.


\section{Existence of Equilibria}
We begin by discussing the existence of equilibria in swap games. The work of 
\citet{Echzell2019dynamics} implies that at least one equilibrium assignment 
exists when the topology is a regular graph. Furthermore, using a potential function 
similar to the one used by \citet{Elkind2019jump} for jump games, we can show that 
equilibria always exist when the topology is a graph of maximum degree $2$; we omit 
the details. Our first result is a proof of non-existence 
of equilibria for every $k \geq 2$ for general topologies.

\begin{figure}[t]
\centering
\includegraphics[scale=0.4]{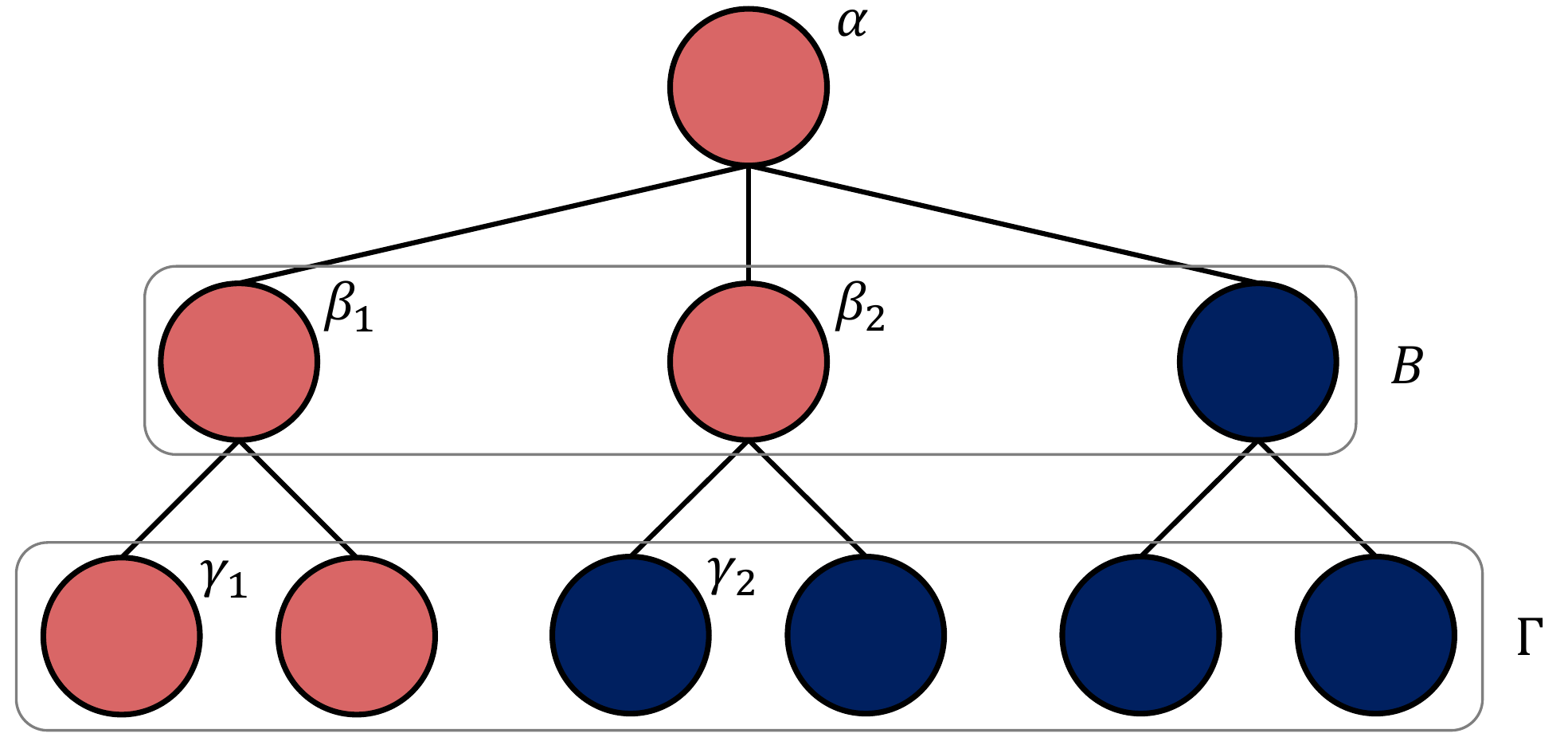}
\caption{The topology of the $2$-swap game considered 
in the proof of Theorem~\ref{thm:non-existence}, 
and an assignment that corresponds to the last case in the analysis; 
it is not an equilibrium since the red agent at node $\alpha$ and
the blue agent at node $\gamma_2$ would like to swap.}
\label{fig:non-existence}
\end{figure}

\begin{theorem}\label{thm:non-existence}
For every $k \geq 2$, there exists a $k$-swap game that does not admit an equilibrium 
assignment, even when all agents are strategic and the topology is a tree.
\end{theorem}

\begin{proof}
We start with the easiest case of $k=2$. 
Consider a $2$-swap game with $10$ strategic agents: $5$ red agents and $5$ blue agents. 
The topology is a tree with a root node $\alpha$, which has three children nodes (set $B$), 
each of which has two children of its own (set $\Gamma$); see Figure~\ref{fig:non-existence}.
Suppose for the sake of contradiction that this game admits an equilibrium assignment $\vv$. 

Since $|B|=3$ and there are only two types of agents, 
at least two nodes in $B$, say $\beta_1$ and $\beta_2$, 
must be occupied by agents of the same type, say red.
Now assume that nodes $\gamma_1$ (a child of $\beta_1$) and 
$\gamma_2$ (a child of $\beta_2$) are occupied by blue agents. 
Then the red agent $\pi_{\beta_1}(\vv)$ and the blue agent 
$\pi_{\gamma_2}(\vv)$ can swap positions 
to increase their utility from strictly smaller than 
$1$ and $0$ to $1$ and positive, respectively. 
Therefore, for at least one of these nodes (say, $\beta_1$)
it must be the case that both of its children 
are occupied by red agents; as there are only five red agents, 
it follows that at least one of the children of $\beta_2$, 
say $\gamma_2$, is occupied by a blue agent. 

If node $\alpha$ is occupied by a blue agent, then the red agent 
$\pi_{\beta_1}(\vv)$ and the blue agent $\pi_{\gamma_2}(\vv)$ 
can both increase their utility by swapping. 
Hence, node $\alpha$ must be occupied by a red agent 
(see Figure~\ref{fig:non-existence}). 
However, this assignment is not an equilibrium either, 
since the red agent $\pi_\alpha(\vv)$ and the blue agent $\pi_{\gamma_2}(\vv)$ 
have an incentive to swap.

For $k \geq 3$, consider a $k$-swap game with $n=k(k^2-2)$ agents such that there are $k^2-2$ agents of type $T_\ell$, for every $\ell \in [k]$. The topology is a tree whose nodes are distributed over three layers, just like in the case $k=2$. Specifically, there is a root node $\alpha$, which has a set $B$ of $k(k-1)-1$ children. Each node in $\beta \in B$ has a set $\Gamma_\beta$ of $k$ children leaf nodes; let $\Gamma = \bigcup_{\beta \in B} \Gamma_\beta$. Next, we will argue about the structure of assignments that cannot be equilibria. 

\begin{lemma}\label{lem:non-existence-1}
An assignment $\vv$ according to which any two nodes $\beta_1, \beta_2 \in B$ are occupied by agents of the same type $T_x$ cannot be an equilibrium in case agents of some type $T_y$, $y \neq x$ simultaneously occupy nodes of $\Gamma_{\beta_1}$ and $\Gamma_{\beta_2}$.
\end{lemma}

\begin{proof}
Let $\vv$ be an assignment according to which nodes $\beta_1, \beta_2 \in B$ are occupied by agents of type $T_x$, and there exist nodes $\gamma_1 \in \Gamma_{\beta_1}$ and $\gamma_2 \in \Gamma_{\beta_2}$, which are occupied by agents of some type $T_y$, with $y \neq x$. Clearly, agent $\pi_{\beta_1}(\vv)$ has utility strictly less than $1$, while agent $\pi_{\gamma_2}(\vv)$ has utility $0$. Therefore, they would like to swap positions in order to increase their utility to $1$ and positive, respectively.
\end{proof}

\begin{lemma}\label{lem:non-existence-2}
An assignment $\vv$ according to which agents of every type occupy nodes of $B$ cannot be an equilibrium. 
\end{lemma}

\begin{proof}
Let $\vv$ be an assignment according to which at least one agent of every type occupies some node of $B$. Without loss of generality, assume that the agent $\pi_\alpha(\vv)$ is of type $T_x$. We now distinguish between the following two cases.
\begin{itemize}
\item An agent of type $T_x$ located at node $\beta \in B$ has a neighbor of type $T_y$, $y \neq x$ located at some node $\gamma \in \Gamma_\beta$. Then, by the assumption of the lemma, there exists at least one agent of type $T_y$ located at some node $\beta' \in B \setminus \{\beta\}$, and therefore agents $\pi_\alpha(\vv)$ and $\pi_\gamma(\vv)$ would like to swap positions in order to increase their utility from strictly less than $1$ and $0$ to $1$ and positive, respectively. 

\item For every agent of type $T_x$ located at some node $\beta \in B$, all agents occupying the nodes of $\Gamma_\beta$ are of type $T_x$. Since $\alpha$ is occupied by an agent of type $T_x$, there are $k^2-3 = (k-1)(k+1)-2$ other agents of type $T_x$ that can completely fill up at most $k-2$ subtrees of $\alpha$ (since each of them consists of $k+1$ nodes). Consequently, there are at least $k-1$ agents of type $T_x$ located at leaf nodes in other subtrees of $\alpha$.

Now, assume that one of these agents of type $T_x$ occupies a node $\gamma \in \Gamma_\beta$ such that $\beta \in B$ is occupied by an agent of type $T_y$, with $y \neq x$. We will now argue that there must exist another agent of type $T_y$ located at some node $\beta' \in B \setminus \{\beta\}$. Assume otherwise that there is no such agent. Then, the remaining $k(k-1)-2$ agents of type $T_y$ occupy leaf nodes of the tree. By Lemma~\ref{lem:non-existence-1}, such agents located in different subtrees of $\alpha$ cannot be connected to agents of the same type. Hence, to cover all  
these $k(k-1)-2$ agents of type $T_y$, agents of $k-1$ different types need to occupy the root nodes of the corresponding subtrees of $\alpha$. However, there are only $k-2$ types left (different than $T_x$ and $T_y$). Consequently, there must exist another agent of type $T_y$ that occupies some  
node $\beta' \in B \setminus \{\beta\}$. As a result, agents $\pi_\gamma(\vv)$ and $\pi_{\beta'}(\vv)$ have incentive to swap positions in order to increase their utility from $0$ and strictly less than $1$ to positive and $1$, respectively. 
\end{itemize}
This completes the proof of the lemma.
\end{proof}

\begin{lemma}\label{lem:non-existence-3}
An assignment according to which there exists a type $T_\ell$ such that no agent of this type appears at nodes of $B$ cannot be an equilibrium.
\end{lemma}

\begin{proof}
Let $\vv$ be an assignment according to which no agent of type $T_\ell$ appears at the nodes of $B$. We first deal with the case $k \geq 4$. Observe that there are at least $k^2-3$ agents of type $T_\ell$ that must occupy nodes of $\Gamma$; one agent of type $T_\ell$ may occupy $\alpha$. By Lemma~\ref{lem:non-existence-1}, agents of the same type that are located in different subtrees of $\alpha$ cannot be connected to agents of the same type. Hence, to cover all these $k^2-3$ agents of type $T_\ell$, agents of $k$ different types must occupy the root nodes of the corresponding subtrees of $\alpha$, which is impossible. 

For $k=3$, if $\alpha$ is not occupied by an agent of type $T_\ell$, then all $k^2-2 = 7$ agents of this type must occupy nodes of $\Gamma$, and the same argument as above leads to a contradiction. Hence, assume that $\pi_\alpha(\vv)$ is of type $T_\ell$. Since $|B|=5$ and no agent of type $T_\ell$ appears at the nodes of $B$, at least three nodes of $B$ must be occupied by agents of the same type. 
Let $B = \{\beta_1, ..., \beta_5\}$, and assume that nodes $\beta_1, \beta_2$ and $\beta_3$ are occupied by agents of type $T_0$. If an agent of type $T_\ell$ occupies some node $\gamma \in \Gamma_\beta$ for any $\beta \in \{\beta_1, \beta_2, \beta_3\}$, then agent $\pi_{\beta'}(\vv)$, $\beta' \in \{\beta_1, \beta_2, \beta_3\} \setminus \{ \beta \}$ has incentive to swap positions with agent $\pi_\gamma(\vv)$ to increase both of their utilities from strictly less than $1$ and $0$ to $1$ and positive. Hence, no agent of type $T_\ell$ can be located at the nodes of $\Gamma_{\beta_1} \cup \Gamma_{\beta_2} \cup \Gamma_{\beta_3}$. Clearly, if agent $\pi_{\beta_4}(\vv)$ or $\pi_{\beta_5}(\vv)$ is of type $T$, for the same reason, no agent of type $T_\ell$ can be located at the corresponding leaf nodes. Hence, both $\beta_4$ and $\beta_5$ must be occupied by agents of the third type $T_1$ and all leaf nodes $\Gamma_{\beta_4} \cup \Gamma_{\beta_5}$ must be occupied by the remaining $6$ agents of type $T_\ell$. However, this clearly cannot be an equilibrium assignment, since agent $\pi_{\beta_4}(\vv)$ would like to swap with any agent in $\Gamma_{\beta_5}$. 
\end{proof}

By Lemmas~\ref{lem:non-existence-2} and~\ref{lem:non-existence-3}, we conclude that no assignment can be an equilibrium.  
\end{proof}

The topology used in the proof of Theorem~\ref{thm:non-existence} for the case $k=2$ is utilized as 
a subgraph in the proof of the following theorem, to show that the problem of deciding 
whether an equilibrium exists is computationally hard.

\begin{theorem}
For every $k \geq 2$, it is NP-complete to decide whether a given $k$-swap game 
admits an equilibrium.
\end{theorem}

\begin{proof}
Membership in NP is immediate: we can verify 
whether a given assignment is an equilibrium 
by simply checking if there exists a pair of agents that would like to swap positions. 
To prove NP-hardness, we give a reduction from the {\sc Clique} problem, 
which in known to be NP-complete. 
An instance of {\sc Clique} consists of an undirected connected graph $H=(X,Y)$ 
and an integer $\lambda$; it is a yes-instance if $H$ contains a complete subgraph 
of size $\lambda$. Without loss of generality, we assume that $\lambda > 5$.

Given an instance $\langle H, \lambda \rangle$ of {\sc Clique}
with $H=(X, Y)$, we will construct a $2$-swap game as follows 
(the reduction can be extended to any $k> 2$ by adding stubborn agents of different types).
Let $d_v$ denote the degree of node $v$ in $H$, 
and set $d_H = \max_{v \in X} d_v$. 
\begin{itemize}
\item There are $\lambda$ strategic red agents and $t=|X| + 5$ strategic blue agents; 
all other agents are stubborn, and will be defined in conjunction with the topology.
\item The topology $G=(V,E)$ consists of three components $G_1$, $G_2$ and $G_3$. 
These are connected to each other via stubborn agents, and are defined as follows:
\begin{itemize}
\item
To define $G_1 = (V_1, E_1)$, let $W_v$ be a set of $2d_H - d_v + 2\lambda - 3$ nodes 
for each $v \in V$. Then, 
$V_1 = X \bigcup_{v\in V} W_v$
and 
$E_1 = Y \cup \{ \{v, w\}: v \in X, w \in W_v \}$.
For every $v \in X$, $d_H$ nodes of $W_v$ are occupied by stubborn red agents, 
while the remaining $d_H - d_v + 2\lambda - 3$ nodes are occupied by stubborn blue agents. 
Observe that every node of $G_1$ has degree $d_1 = 2d_H + 2\lambda - 3$.

\item
The subgraph $G_2 = (A \cup B, E_2)$ is a complete bipartite graph with 
$|A|=\lambda-5$ and $|B|=4d_1$. 
Out of the $4d_1$ nodes of $B$, 
$2d_1 + 1$ nodes are occupied by stubborn red agents, while the remaining
$2d_1 - 1$ nodes are occupied by stubborn blue agents.

Hence, a strategic red agent occupying a node of $A$ has utility 
$\chi_r = \frac{2 d_1 + 1}{4 d_1} = \frac{1}{2} + \frac{1}{4d_1}$. 
Similarly, a strategic blue agent has utility 
$\chi_b = \frac{2 d_1 - 1}{4 d_1} = \frac{1}{2} - \frac{1}{4d_1}$.

\item
To define $G_3=(V_3, E_3)$, let $(V'_3, E'_3)$ be the graph used in the proof of 
Theorem~\ref{thm:non-existence}, for which there is no equilibrium assignment; 
see Figure~\ref{fig:non-existence}. For every node $v \in V_3'$ of degree $3$, 
let $Z_v$ be a set of $100d_1$ nodes such that $50 d_1$ of these nodes 
are occupied by stubborn red agents, while the remaining $50 d_1$ nodes 
are occupied by stubborn blue agents. For every $v \in V_3'$ of degree $1$, 
let $Z_v$ be a set of $10d_1$ nodes such that $5d_1$ of these nodes 
are occupied by stubborn red agents, while the remaining $5d_1$ nodes 
are occupied by stubborn blue agents. Then, $V_3 = V'_3 \bigcup_{v \in V_3'} Z_v $ 
and $E_3 = E'_3 \cup \{ \{v, w\}: v \in V'_3, w \in Z_v\}$. 

One can easily verify that the utility of a strategic agent (red or blue) occupying a node 
of $G_3$ is at least $\psi_0 = \frac{5 d_1 -1}{10 d_1+1} > \frac{1}{2} - 
\frac{1}{4d_1}$ and at most $\psi_1 = \frac{5 d_1 + 1}{10 d_1+1} < \frac{1}{2} + \frac{1}{4d_1}$.
\end{itemize}
\end{itemize}

Now, assume that $H$ has a clique of size $\lambda$, and let $\vv$ be the assignment 
in which the strategic red 
agents occupy the nodes of the clique, and the strategic 
blue agents occupy the remaining nodes.
Each strategic red agent is connected to $\lambda-1 + d_H$ other red agents (strategic and 
stubborn) in $G_1$, and thus has utility
$$
u = \frac{ \lambda-1 + d_H} { d_1} = 
\frac{ d_H + \lambda - 1.5 + 0.5} { 2d_H + 2\lambda - 3 } \ge 
\frac{1}{2} + \frac{1}{2 d_1}.
$$
Clearly, since $u > \chi_r$ and $u > \psi_1$, no strategic red agent would be willing to swap positions with another strategic agent in $G_2$ or $G_3$. By swapping positions with a blue agent within $G_1$, a strategic red agent would still have at most $\lambda-1+d_H$ friends, and since every node in $G_1$ has the same degree, her utility cannot be improved. Hence, no strategic red agent has a profitable deviation, and $\vv$ is an equilibrium.

Conversely, assume that $H$ does not contain a clique of size $\lambda$, and for the sake of contradiction also assume that there is an equilibrium assignment $\vv$.

Suppose that some strategic red agents are located in $G_1$. 
It cannot be the case that each of them is adjacent to $\lambda-1$ other strategic red agents, 
as this would mean that the nodes they occupy form a clique of size $\lambda$. 
Hence, at least one of these agents, say agent $i$, is adjacent to at most $\lambda-2$ strategic red agents. 
Since every node of $G_1$ has degree $d_1$ and every node is adjacent to $d_H$ stubborn red agents, 
the utility of $i$ is 
$$
u_i \leq \frac{d_H + \lambda-2}{d_1} = 
\frac{ d_H + \lambda - 1.5 - 0.5} { 2d_H + 2\lambda - 3 } = 
\frac{1}{2} - \frac{1}{2 d_1}.
$$ 
We have that $u_i < \chi_r$ and $u_i < \psi_0$, and hence agent $i$ has incentive to move to $G_2$ or $G_3$.
On the other hand, the utility that a strategic blue agent $j$ that is currently located in 
$G_2$ or $G_3$ can obtain by swapping positions with $i$ is
$$
u_j = 1 - u_i \geq \frac{1}{2} + \frac{1}{2 d_1}.
$$
Since $u_j > \chi_b$ and $u_j > \psi_1$, agent $j$ also has an incentive to swap positions 
with agent $i$, and hence $\vv$ cannot be an equilibrium assignment. 
Therefore, no strategic red agent is located in $G_1$.

Similarly, observe that $\chi_r > \psi_1$ and $\chi_b < \psi_0$, meaning that 
strategic red agents would prefer to be in $G_2$, while strategic blue agents 
would prefer to be in $G_3$. Thus, for $\vv$ to be an equilibrium assignment, 
it must be the case that all if a node of $G_2$ is not occupied by a stubborn agent, 
it is occupied by a strategic red agent. As a result, there are $5$ strategic red 
and $5$ strategic blue agents in $G_3$. 
However, similarly to the proof of Theorem~\ref{thm:non-existence}, 
we can argue that there is no equilibrium assignment for these agents in $G_3$; we omit the details here.
Since we have exhausted all possibilities, it follows that if $H$ does not have a clique
of size $\lambda$, then there is no equilibrium assignment.
\end{proof}


\section{Social Welfare}
Here, we consider the efficiency of equilibrium assignments in terms of social welfare, and bound the social price of anarchy and stability for many interesting cases. We restrict our attention to games consisting of only strategic agents and such that there are at least two agents per type. Swap games with stubborn agents or strategic agents that are unique of their type can easily be seen to have unbounded social price of anarchy.\footnote{For any $k\geq 2$, consider a $k$-swap game with a star topology and $k$ types of agents such that there is only one red strategic agent, while the other types consist of at least two strategic agents and of some stubborn ones located at peripheral nodes. The assignment according to which the red agent occupies the center node is an equilibrium with $0$ social welfare, while any assignment such that the center node is occupied by a non-red agent has positive social welfare.}

We start with the social price of anarchy of $2$-swap games, and consider the general case (given the above restrictions) and the case where each type consists of the same number of agents. 

\begin{figure}[t]
\centering
\includegraphics[scale=0.35]{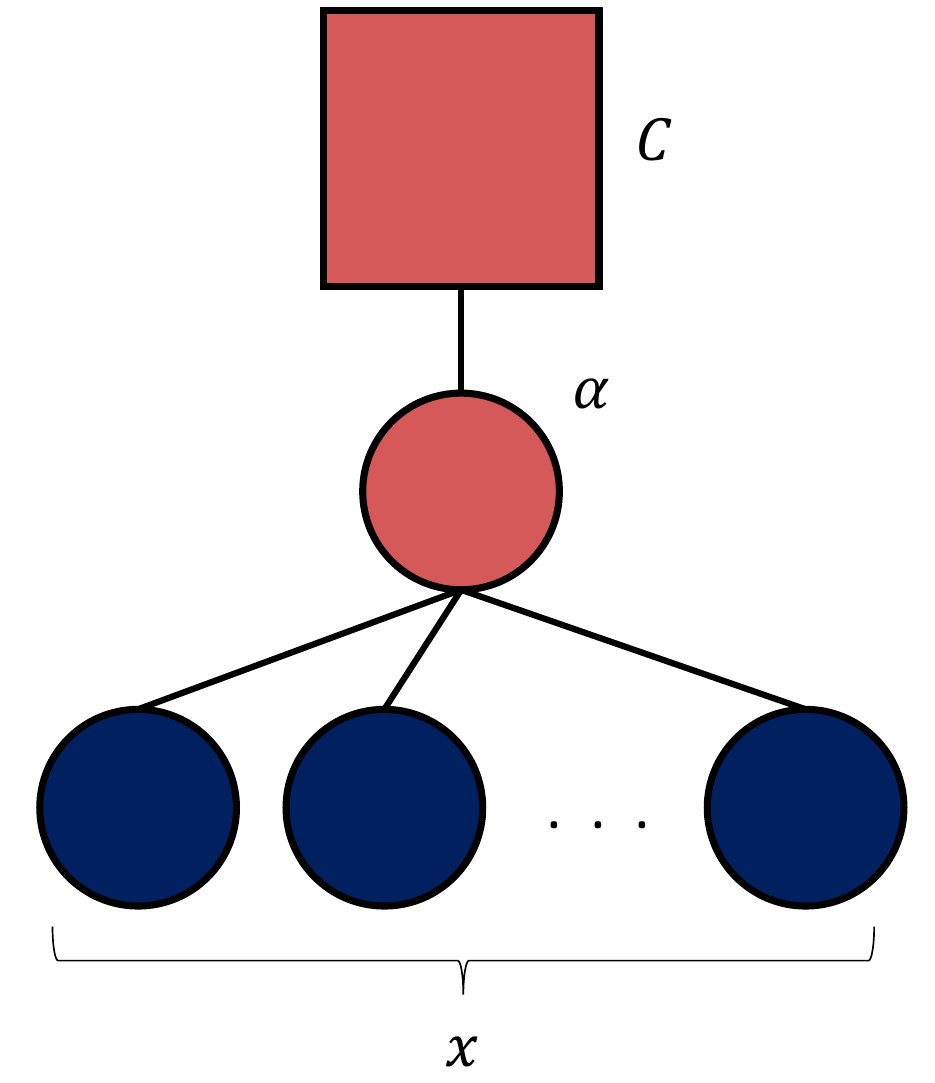}
\caption{The topology and the equilibrium assignment of the lower bound instance 
in the proof of part (i) of Theorem~\ref{thm:poa-2-types}. The big red square 
represents a clique whose nodes are occupied by red agents only.}
\label{fig:poa-2-types}
\end{figure}

\begin{theorem}\label{thm:poa-2-types}
The social price of anarchy of $2$-swap games with only strategic agents is
\begin{itemize}
\item[(i)] $\Theta(n)$ if there are at least two agents of each type, and
\item[(ii)] between $921/448 \approx 2.0558$ and $4$ if each type consists of the same number of agents.
\end{itemize}
\end{theorem}

\begin{proof}
We prove each part separately.

\medskip

\noindent
{\bf Part(i).} 
For the lower bound, consider a $2$-swap game with $n-2$ red and $2$ blue agents who are to be positioned on the nodes of a star topology. Then, the assignment according to which one of the blue agents occupies the central node is an equilibrium with social welfare $1 + \frac{1}{n-1} \leq 2$, while the optimal assignment is such that the central node is occupied by a red agent for a social welfare of $n-3 + \frac{n-3}{n-1} \geq n-3$.

For the upper bound, consider any $2$-swap game with $n$ agents such that there are $n_r \geq 2$ red and $n_b$ blue agents, and let $\vv$ be any equilibrium assignment of this game. If there is any red agent $\ell$ with zero utility in $\vv$, then it cannot be the case that this agent is connected to all blue agents. If this were the case, then since there is another red agent and the graph is connected, at least one blue agent must be connected to this red agent, get utility strictly less than $1$, and have incentive to swap positions with agent $\ell$ so that they both increase their utility. Hence, in order for $\vv$ to be an equilibrium in the presence of $\ell$ getting zero utility, any blue agent not connected directly to $\ell$ must get utility $1$ so that she does not want to swap with $\ell$. Consequently, $\SW(\vv) \geq 1$. In the case where every agent has positive utility in $\vv$, since the graph is connected, it must be the case that every agent get utility at least $1/n$, and therefore again $\SW(\vv) \geq 1$. Now the bound follows since the optimal social welfare is at most $n$.

\medskip

\noindent
{\bf Part (ii).}
For the lower bound, consider a 
$2$-swap game with the following topology: 
there is a node $\alpha$ of degree $x+1$ that is 
connected to $x$ leaf nodes and to one node in a clique $C$ of size $x-1$. There is an equilibrium 
$\vv$ where $\alpha$ is occupied by a red agent $r$, all leaf nodes are occupied by blue agents, 
and all nodes of $C$ are occupied by red agents; see Figure~\ref{fig:poa-2-types}. Hence,
$$\SW(\vv) = x-1 + \frac{1}{x+1} = \frac{x^2}{x+1}.$$
On the other hand, for the assignment $\vv^*$ obtained from $\vv$ by swapping $r$ 
with one of the blue agents we have
$$\SW(\vv^*) = 2x-3 + \frac{x-2}{x-1} + \frac{x-1}{x+1}.$$
Hence,  the price of anarchy is at least
$$\frac{2x^3-x^2-5x+2}{x^2(x-1)},$$
an expression that takes it maximum value $667/324 \approx 2.05864$ at $x=9$.

For the upper bound, consider a $2$-swap game with $n=2x$ agents such that there are 
$x \geq 2$ red and $x$ blue agents. First, assume that some agents get zero utility in the 
equilibrium assignment $\vv$. Observe that it cannot be the case that there exist agents of both 
types who have zero utility in $\vv$. Indeed, if this was true for a non-adjacent red-blue pair, then 
these agents would have an incentive to swap and increase their utility from zero to $1$. 
On the other hand, suppose this is true for an adjacent red-blue pair $(r,b)$.
If both of $r$ and $b$ have other neighbors (besides $r$ and $b$), 
then by swapping they can increase their utility from $0$ to positive. 
Hence, suppose that $r$ occupies a leaf node and is connected only to $b$. 
Then, since the graph is connected, there must exist a blue agent $b' \neq b$ with utility in $(0,1)$ who would like to swap with $r$ to increase both of their utilities from $0$ to positive for $r$ and from strictly less than $1$ to $1$ for $b'$. 

Thus, assume that at least one blue agent has zero utility and all red agents have positive 
utility. We denote by $B_0$ the set of blue agents with zero utility, by $R_1$ the set of red agents 
with utility $1$, and by $R_<$ the set of red agents whose utility is strictly less than $1$. We 
have $|R_1|+|R_<|=x$, and each agent in $B_0$ is connected to all agents in $R_<$; otherwise 
a non-adjacent pair of agents $i \in R_<$, $j \in B_0$ would like to swap.
If $|R_<|=1$, then $|R_1|=x-1$. The utility of the unique agent $i \in R_<$ is 
$u_i(\vv) \geq \frac{1}{1+x}$, and thus 
$$\SW(\vv) \geq x-1+\frac{1}{1+x}.$$ 
Since the optimal social welfare is at most 
$n=2x$ and $x \geq 2$, the price of anarchy is at most 
$$\frac{2(1+x)}{x} = 2 + \frac{2}{x} \leq 3.$$ 
Now, assume that $|R_<|\geq 2$. If $|B_0|\geq 2$, then any agent in $B_0$ 
would like to swap with any agent in $R_<$ to 
get positive utility (since each agent in $R_<$ is connected to all agents in $B_0$). Thus, 
in $\vv$ no agent $i \in R_<$ wants to swap with any agent in $B_0$. Since each agent in 
$B_0$ is connected to all agents in $R_<$ and no other agent, the utility that agent $i$ would get 
by agreeing to swap is $u_i(\vv) = \frac{|R_<|-1}{|R_<|} \geq 1/2$, yielding 
$$\SW(\vv) \geq |R_1|+|R_>|/2 \geq x/2.$$ 
Since the maximum welfare is $2x$, 
we can upper-bound the price of anarchy by $4$. 
If $|B_0|=1$, then if $i$ is connected to another blue agent (not in $B_0$), 
her utility is at least $1/2$ for the same reason as before. Otherwise, $i$ is connected 
to red agents only and the one agent in $B_0$, so $u_i(\vv) \geq 1/2$ (only one red agent in 
the worst case), which again yields an upper bound of $4$ on the price of anarchy.

Next, we assume that all agents have positive utility. 
Since $\vv$ is an equilibrium, for every red-blue pair of agents it holds that
at least one of them has no incentive to swap positions. 
Let $(i,j)$ be a red-blue pair of agents, and assume that $i$ does not want to swap. 
If $i$ and $j$ are {\em not} neighbors in $\vv$, then it must be that
$u_i(\vv) \geq 1 - u_j(\vv)$ and hence  $u_i(\vv) + u_j(\vv) \geq 1.$
Otherwise, $i$ and $j$ are neighbors in $\vv$. 
It may be the case that $u_i(\vv) + u_j(\vv) < 1$, in which case $u_i(\vv) \in (0,1)$ 
and $u_j(\vv) \in (0,1)$. Assume that the blue agent $j$ has $x_r \geq 0$ red neighbors 
besides $i$, and $x_b \geq 1$ blue neighbors. Then, 
$u_j(\vv) = \frac{x_b}{x_r+x_b+1}$, and
\begin{align} \nonumber
u_i(\vv) &\geq \frac{x_r}{x_r+x_b+1} = 1 - u_j(\vv) - \frac{1}{x_r+x_b+1} \\ \label{eq:poa-tight}
&\geq \frac{1}{2} - u_j(\vv) \Leftrightarrow u_i(\vv) + u_j(\vv) \geq \frac{1}{2},
\end{align}
where the inequality follows since $\frac{1}{x_r+x_b+1}$ is decreasing in $x_r \geq 0$ 
and $x_b\geq 1$. Therefore, in any case we have $u_i(\vv) + u_j(\vv) \geq \frac{1}{2}$, 
for every red-blue pair of agents $i$ and $j$.
Since there are $x^2$ distinct red-blue pairs, and each agent participates in exactly $x$ 
such pairs, by summing over all these inequalities, we obtain 
$x \cdot \SW(\vv) \geq \frac{1}{2} x^2$ and therefore $\SW(\vv) \geq \frac{1}{2}x$.
The bound then follows since the maximum possible social welfare is $n=2x$.
\end{proof}

We remark that even though the upper bound of $4$ for the case where each types consists of the same number of agents (and every agent has positive utility at equilibrium) is probably not tight, one cannot expect to improve it using the same technique. In particular, to prove this upper bound, we focused on an arbitrary pair of agents $(i,j)$, and used the equilibrium definition, according to which at least one of these agents does not want to swap positions. Then, we were able to show that $u_i(\vv) + u_j(\vv) \geq 1/2$. We now argue that this inequality is actually tight, and hence to improve the upper bound one needs to argue in more detail about the structure of the equilibrium. Consider a variant of the topology depicted in Figure~\ref{fig:poa-2-types} in which every leaf node is connected to another leaf node. Then, for the depicted assignment $\vv$, the red agent $\pi_\alpha(\vv)$ has utility $\frac{1}{x+1}$, while any blue agent has utility exactly $1/2$. Hence, the sum of the utility of the red agent $\pi_\alpha(\vv)$ and the utility of any blue agent is almost $1/2$ as $x$ becomes large.

We continue by showing that, surprisingly, for three types or more, the social price of anarchy can be unbounded, even for the special case of equal number of agents per type.

\begin{figure}[t]
\centering
\includegraphics[scale=0.4]{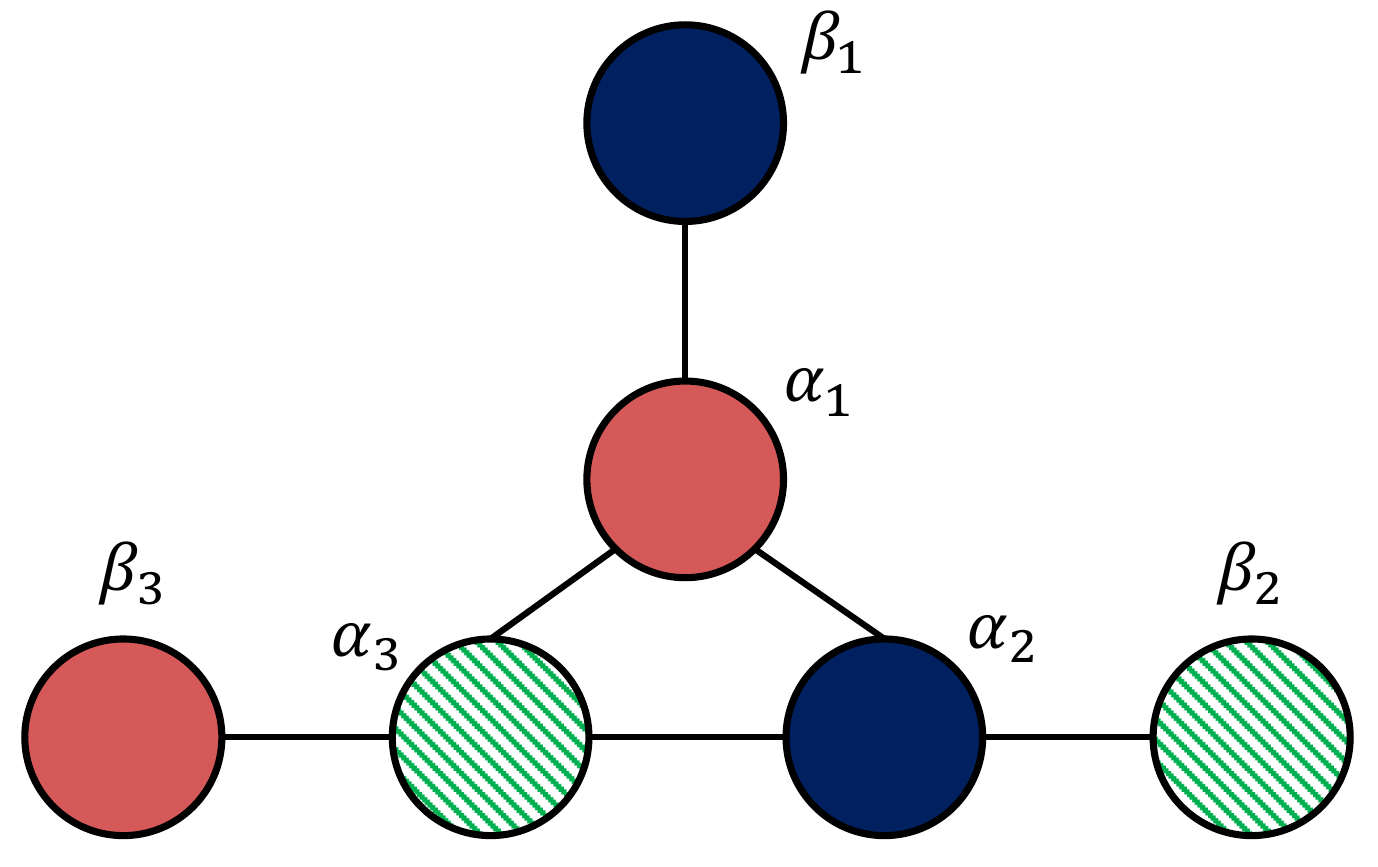}
\caption{The topology and the equilibrium assignment of the $k$-swap game considered in the proof of Theorem~\ref{thm:poa-3-types} for $k=3$. Here, $T_1=\text{red}$, $T_2=\text{blue}$ and $T_3=\text{green}$.}
\label{fig:poa-3-types}
\end{figure}

\begin{theorem}\label{thm:poa-3-types}
For every $k\geq 3$, the social price of anarchy of $k$-swap games can be unbounded, even when there is an equal number of strategic agents per type.
\end{theorem}

\begin{proof}
Consider a $k$-swap game with $n=2k$ agents such that there are exactly two agents of each of the $k \geq 3$ types $T_1, \dots, T_k$. The topology $G$ consists of $k$ nodes $\{\alpha_1, ..., \alpha_k\}$ that form a cycle, and each node $\alpha_\ell$, $\ell \in [k]$ is also connected to an auxiliary node $\beta_\ell$; see Figure~\ref{fig:poa-3-types} for the topology and the equilibrium assignment discussed in the following for $k=3$.

Let $\vv$ be the assignment according to which node $\alpha_\ell$ is occupied by an agent of type $T_\ell$, while node $\beta_\ell$ is occupied by an agent of type $T_{\ell+1}$, where the subscripts are mod $\ell$.
This assignment is clearly an equilibrium since there exists no pair of agents that would like to swap positions, and every agent has zero utility.  
In particular, observe that agent $\pi_{\alpha_{\ell+1}}(\vv)$ of type $T_{\ell+1}$ would like to move only to node $\alpha_\ell$ in order to connect with the agent $\pi_{\beta_\ell}(\vv)$ who is also of type $T_{\ell+1}$. 
However, the agent $\pi_{\alpha_\ell}(\vv)$ of type $T_\ell$ has no incentive to move to node $\alpha_{\ell+1}$ since the other agent of type $T_\ell$ is at node $\beta_{\ell-1}$. 
Now, consider agent $\pi_{\beta_\ell}(\vv)$ of type $T_{\ell+1}$ who is connected only to an agent of type $T_\ell$ located at $\alpha_\ell$. The only agent that would like to swap positions with $\pi_{\beta_\ell}(\vv)$ is $\pi_{\beta_{\ell-1}}(\vv)$ who is of type $T_\ell$. However, such a swap cannot increase the utility of $\pi_{\beta_\ell}(\vv)$ since agent $\pi_{\alpha_{\ell-1}}(\vv)$ is of type $T_{\ell-1} \neq T_{\ell+1}$. Therefore, $\vv$ is an equilibrium with $\SW(\vv)=0$.

On the other hand, consider the assignment $\vv^*$ according to which nodes $\alpha_\ell$ and $\beta_\ell$ are occupied by the two agents of type $T_\ell$, for every $\ell \in [k]$. Since every agent has now positive utility, $\SW(\vv^*) > 0$, and the social price of anarchy is unbounded.
\end{proof}

Next, we turn our attention to the social price of stability and show a constant lower bound for $2$-swap games. 

\begin{figure}[t]
\centering
\includegraphics[scale=0.4]{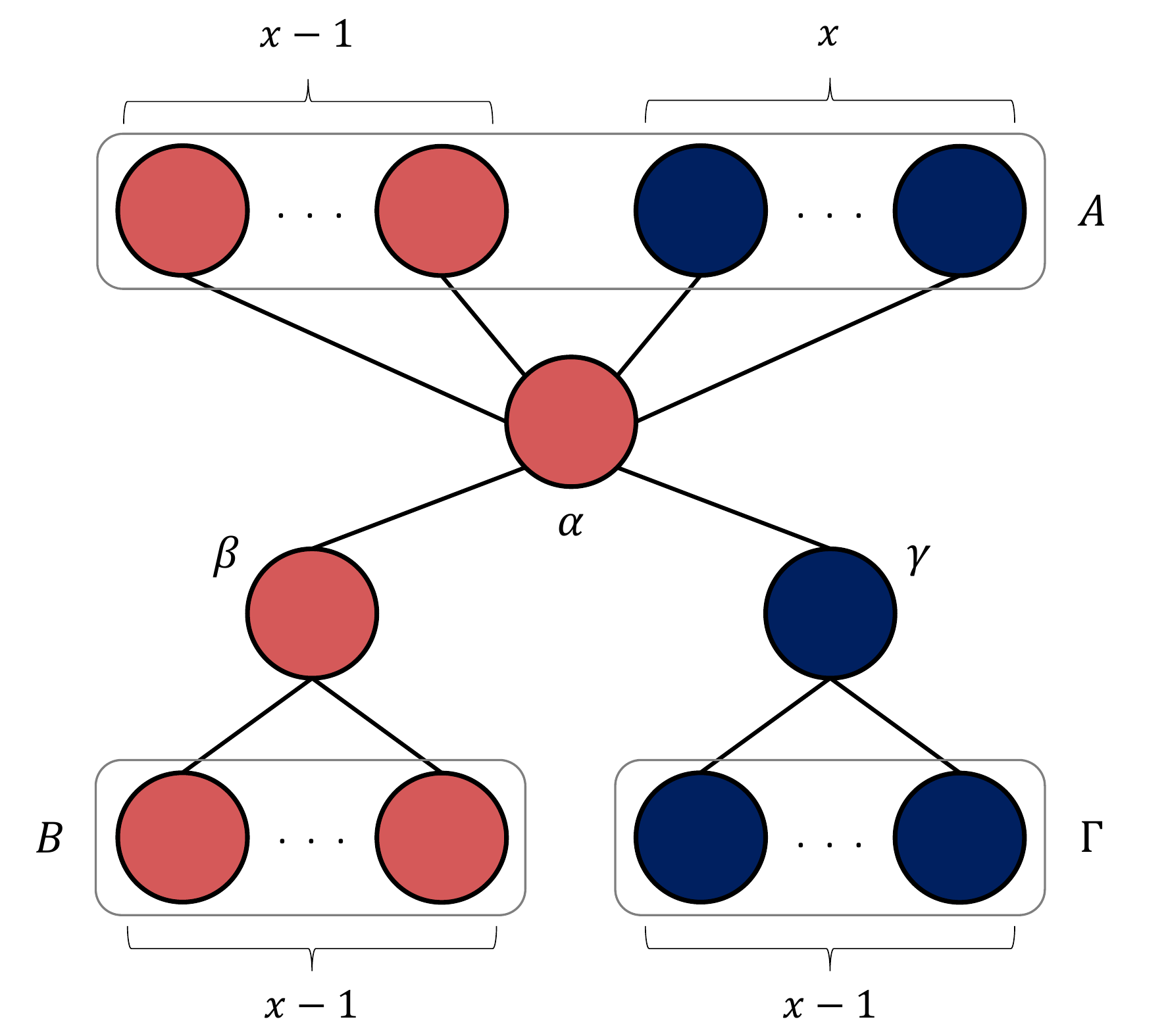}
\caption{The topology and the best equilibrium assignment for the $2$-swap games considered in the proof of Theorem~\ref{thm:pos}.}
\label{fig:pos}
\end{figure}

\begin{theorem}\label{thm:pos}
The social price of stability of $2$-swap games is at least $4/3$.
\end{theorem}

\begin{proof}
Let $x \geq 3$ be a parameter, and consider a $2$-swap game with $2x+1$ red and $2x+1$ blue agents. The topology is a tree with a root node $\alpha$, which is connected to two nodes $\beta$ and $\gamma$, as well as to a set $A$ of $2x-1$ leaf nodes. Moreover, node $\beta$ is connected to a set $B$ of $x$ leaf nodes, and node $\gamma$ is connected to a set $\Gamma$ of $x$ more leaf nodes. The topology and the best equilibrium assignment (which we discuss below) are depicted in Figure~\ref{fig:pos}.

We will now argue about the structure of any equilibrium for this particular swap game. Without loss of generality, we assume that the root node $\alpha$ is occupied by a red agent, and switch between a few cases depending on the number of blue agents that occupy the lead nodes of set $A$ that are directly connected to $\alpha$.

First, assume that there are at least $x+1$ blue agents at the nodes of set $A$. Then, there can be at most $x-2$ red agents at the nodes of $A$, which means that the remaining at least $x+2$ red agents need to occupy nodes of the $\beta$- and $\gamma$-subtrees. Since any of these subtrees have a total of $\ell+1$ nodes, at least one of these red agents, say agent $i$, must be connected to at least one blue agent. Clearly, such an assignment cannot be an equilibrium since agent $i$ and any of the blue agents at the nodes of $A$ have incentive to swap positions to increase their utilities from strictly smaller than $1$ and $0$ to $1$ and positive, respectively.

Second, assume that there are exactly $x$ blue agents at the nodes of set $A$, and hence the remaining $x-1$ nodes of $A$ are occupied by red agents. Then, it is easy to verify that the only equilibrium assignment $\vv_1$ (up to symmetries) is such that all nodes of the $\beta$-subtree are occupied by red agents, and all nodes of the $\gamma$-subtree are occupied by blue agents. The social welfare of this equilibrium is
$$\SW(\vv_1) = 3x + \frac{x}{x+1} + \frac{x}{2x+1} \leq 3x+2.$$

Third, assume that the number of blue agents at the nodes of set $A$ is between $1$ and $x-1$. Then, there are at least $x$ red agents at the nodes of $A$. Since any of the $\beta$- and $\gamma$-subtrees have a total of $x+1$ nodes, at least one of the remaining at most $x$ red agents, say agent $i$, must necessarily be connected to some blue agent. As in the first case, such an assignment cannot be an equilibrium since agent $i$ and any of the blue agents at the nodes of $A$ have incentive to swap positions to increase their utilities from strictly smaller than $1$ and $0$ to $1$ and positive, respectively.

Finally, assume that all nodes of $A$ are occupied by red agents, and there is only one remaining red agent $i$, who will inevitably be connected to some blue agents. No assignment $\vv_2$ according to which $i$ occupies a leaf node of $B$ (or $\Gamma$) can be an equilibrium, since $i$ and the blue agent $\pi_\gamma(\vv_2)$ (or $\pi_\beta(\vv_2)$) have incentive to swap positions and increase their utilities from $0$ and strictly smaller than $1$ to positive and $1$, respectively. Hence, in any equilibrium assignment $\vv_2$, agent $i$ occupies either node $\beta$ or node $\gamma$. The social welfare is
$$\SW(\vv_2) = 3x-1 + \frac{2x}{2x+1} + \frac{1}{x+1} + \frac{x}{x+1} \leq 3x+1.$$

Now consider the assignment $\vv^*$ according to which the red agents occupy node $\alpha$, all nodes of $A$, and one node of $B$, while all other nodes are occupied by blue agents. The social welfare of this assignment is
$$\SW(\vv^*) = 4x-1 + \frac{2x-1}{2x+1} + \frac{x-1}{x+1} \geq 4x-1.$$
Therefore, the social price of anarchy is at least $\frac{4x-1}{3x+2}$, which tends to $4/3$ as $x$ tends to infinity.
\end{proof}

For $k$-swap games with topology that is a $\delta$-regular graph (in which all nodes have degree equal to $\delta$), we show an upper bound of $1$ on the social price of stability by exploiting a potential function, similar to the one defined by \citet{CLM18} and \citet{Echzell2019dynamics} to show the existence of equilibria in such games. 

\begin{theorem}\label{thm:pos-regular}
The social price of stability in $k$-swap games with topology that is a $\delta$-regular graph is $1$.
\end{theorem}

\begin{proof}
\citet{Echzell2019dynamics} showed that for $k$-swap games with a $\delta$-regular topology, 
$$\Phi(\vv) = \sum_{i \in R} |N_i(\vv) \setminus F_i|$$ 
is a potential minimization function. 
Using similar arguments, we can show that the complement,
$$\overline{\Phi}(\vv) = \sum_{i \in R} |N_i(\vv) \cap F_i|,$$
is a potential maximization function. Consider any pair of agents $(i,j)$ such that $i$ is of type $T_x$ and $j$ is of type $T_y$, with $y \neq x$. Since $i$ and $j$ swap positions if and only if they can both increase their utility, and since $N_i(\vv)=N_j(\vv)=\delta$ for any assignment $\vv$, we have that
$$|N_i(\vv) \cap F_i| < |N_i(\vv^{i \leftrightarrow j})\cap F_i|
\text{ \ \  and \ \ } 
|N_j(\vv) \cap F_j| < |N_j(\vv^{i \leftrightarrow j})\cap F_j|$$
Any agent $\ell \in (N_i(\vv) \cap F_i) \cup ( N_j(\vv) \cap F_j)$ has one less friend in $\vv^{i \leftrightarrow j}$ than in $\vv$, and hence
$$|N_\ell(\vv^{i \leftrightarrow j}) \cap F_\ell| = |N_\ell(\vv) \cap F_\ell| - 1.$$
On the other hand, any agent $\ell \in (N_i(\vv^{i \leftrightarrow j})\cap F_i) \cup (N_j(\vv^{i \leftrightarrow j})\cap F_j)$ has one more friend in $\vv^{i \leftrightarrow j}$ than in $\vv$, and hence
$$|N_\ell(\vv^{i \leftrightarrow j}) \cap F_\ell| = |N_\ell(\vv) \cap F_\ell| + 1.$$
For any other agent, the friends they have as neighbors have not changed. 
Therefore, we can now easily see that 
$$\overline{\Phi}(\vv^{i \leftrightarrow j}) - \overline{\Phi}(\vv) > 0,$$ 
as desired.

Now, observe that by the definition of the utility of each strategic agent and the fact that the topology is $\delta$-regular, for any assignment $\vv$, we have that
\begin{align}  \label{eq:potential-sw}
\SW(\vv) &= \sum_{i \in R} u_i(\vv) = \sum_{i \in R} \frac{|N_i(\vv) \cap F_i|}{|N_i(\vv)|} 
= \frac{1}{\delta} \cdot \sum_{i \in R} |N_i(\vv) \cap F_i| = \frac{1}{\delta} \cdot \overline{\Phi}(\vv).
\end{align}
Let $\vv^*$ be an optimal assignment. If $\vv^*$ is an equilibrium, then the social price of stability is $1$. Otherwise, we let the strategic agents play and swap positions until they reach an equilibrium $\vv$. Since $\overline{\Phi}$ is a potential maximization function, we have that $\overline{\Phi}(\vv) \geq \overline{\Phi}(\vv^*)$, and by \eqref{eq:potential-sw}, we obtain
\begin{align*}
\SW(\vv) = \frac{1}{\delta} \cdot \overline{\Phi}(\vv) 
\geq \frac{1}{\delta} \cdot \overline{\Phi}(\vv^*) = \SW(\vv^*),
\end{align*}
and the bound follows by rearranging terms.
\end{proof}

Next, we focus on the problem of computing assignments of high social welfare. Observe that whether the agents are allowed to pairwise swap positions or jump to empty nodes has no effect in the complexity of this problem, and hence we already known that it is NP-complete by the work of \citet{Elkind2019jump}. 
However, one of their main assumptions is that the topology is a graph with strictly more nodes than agents (so that there are empty nodes where the agents can jump to). Consequently, their proof does not cover our case, where the topology consists of a number of nodes that is exactly equal to the number of agents. The proof of our next theorem is fundamentally different and subsumes that of \citet{Elkind2019jump} for $k \geq 3$; for $k=2$, we were unable to prove the hardness of the problem. 

\begin{theorem}
For every $k \geq 3$, given a rational number $\xi$, it is NP-complete 
to decide whether there exists an allocation that has social welfare at least $\xi$.
\end{theorem}

\begin{proof}
Membership in NP is immediate: given an assignment, we can sum up the utilities of the 
strategic agents and check whether the social welfare is at least $\xi$.
To prove NP-hardness, we give a reduction from an NP-complete variant of 
the min-cut problem with additional cardinality constraints on the size of the subsets, 
to which we refer as the {\sc Equal-Min-Cut} problem~\cite{G74}.
An instance of {\sc Equal-Min-Cut} consists of a graph $H=(X,Y)$, 
two distinguished nodes $s,t \in X$, and an integer $W$. 
It is a yes-instance if and only if there exist disjoint subsets 
of nodes $X_1$ and $X_2$ such that $X_1 \cup X_2 = X$, $|X_1| = |X_2|$, 
$s\in X_1$, $t \in X_2$ and $|\{ \{v,z\} \in Y: v \in X_1, z \in X_2|\}\leq W$. 
Without loss of generality, we assume that $|X|$ is an even number, and by convention we denote an edge $\{v,z\}$ as $vz$ to simplify our notation. 

Given an instance $\langle H, s, t, W\rangle$ of {\sc Equal-Min-Cut} with $H=(X,Y)$, 
we construct an instance of our social welfare maximization problem as follows:
\begin{itemize}
\item
There are $|X|/2-1$ strategic red and $|X|/2-1$ strategic blue agents.
\item
The topology $G=(V,E)$ consists of $H$ with additional nodes and edges.
Let $s_0$ and $t_0$ be two auxiliary nodes, and define 
$X_0 = \{s,s_0,t,t_0\}$ and 
$Y_0 = Y \cup \{ s_0v: sv \in Y \} \cup \{ t_0v: tv \in Y\}$. 
Let $d_v = |e \in Y_0: v \in e|$ for every $v \in X \setminus X_0$, and $d_0 = \max_{v\in X \setminus X_0} d_v$.
Also, let $Z_v$ be a set of $d_0-d_v$ nodes for every $v \in X\setminus X_0$. 
Then, $G$ is such that 
$V = X \cup \{s_0,t_0\} \bigcup_{v \in  X_0} Z_v$ and 
$E = Y_0 \cup \{ vz: v \in  X\setminus X_0, z \in Z_v \}$.
Observe that in $G$, every node $v \in X \setminus X_0$ has degree exactly equal to $d_0$. 
\item 
The nodes $s$ and $s_0$ are occupied by stubborn red agents, 
$t$ and $t_0$ are occupied by stubborn blue agents, 
and all nodes in $\bigcup_{v \in X \setminus X_0} Z_v$ are occupied by stubborn green agents.
\item 
Finally, let $\xi = \frac{2}{d_0} ( |Y| - W)$.
\end{itemize}

For any assignment $\vv$ and node $v \in V$, let $\chi_v(\vv)$ denote the type of the agent 
occupying node $v$.
We will show that the social welfare of $\vv$ is decreasing in the number of edges of $Y$ 
that are occupied by agents of different types. We have
\begin{align*}
\SW(\vv) &= \sum_{v \in X \setminus X_0} \frac{|\{ vz \in Y_0: \chi_v(\vv) = \chi_z(\vv)\}|}{d_0}\\
&= \frac{1}{d_0} \sum_{v \in X \setminus X_0} {|\{ vz \in Y_0: \chi_v(\vv) = \chi_z(\vv), 
z \notin X_0\}|} \\ 
&\quad + \frac{1}{d_0}  \sum_{v \in X \setminus X_0} {|\{vz \in Y_0: \chi_v(\vv) = \chi_z(\vv), 
z \in X_0\}|}.
\end{align*}
Since $\chi_s(\vv) = \chi_{s_0}(\vv)$, $\chi_t(\vv) = \chi_{t_0}(\vv)$, 
$vs \in Y \Leftrightarrow vs_0 \in Y_0$, 
and $vt \in Y \Leftrightarrow vt_0 \in Y'$, 
it follows that for every $v \in X \setminus X_0$, 
\begin{align*}
|\{vz \in Y_0: \chi_v(\vv) = \chi_z(\vv), z \in X_0\}|
&= 2 |\{vz \in Y_0: \chi_v(\vv)\ = \chi_z(\vv), z \in \{s,t\}\}| \\
&= 2 |\{vw \in Y:  \chi_v(\vv)\ = \chi_z(\vv), z \in \{s,t\}\}|.
\end{align*}
Furthermore, 
\begin{align*}
|\{vz \in Y_0: \chi_v(\vv) = \chi_z(\vv), z \not\in X_0\}| = |\{vw \in Y_0: \chi_v(\vv) = \chi_z(\vv), z \not\in \{s,t\}\}|
\end{align*}
Therefore, we obtain
\begin{align} \nonumber
\SW(\vv) &= \frac{2}{d_0} |\{vz \in Y: \chi_v(\vv) = \chi_z(\vv), 
z \in X \setminus \{s,t\}\}| \nonumber \\
&\quad +\frac{2}{d_0} |\{ vz \in Y: \chi_v(\vv)\ = \chi_z(\vv), z \in \{s,t\}\}| \nonumber \\
&= \frac{2}{d_0} |\{ vz \in Y: \chi_v(\vv) = \chi_z(\vv)\}| \nonumber \\
&= \frac{2}{d_0} \Big( |Y| - |\{ vz \in Y: \chi_v(\vv) \neq \chi_z(\vv)\}| \Big). \label{eq:SW-equal-min-cut}
\end{align}

Now, assume that the input instance of {\sc Equal-Min-Cut} is a yes-instance, 
and let $X = X_1 \cup X_2$ be the satisfying partition. 
Let $\vv$ be such that the strategic red agents occupy the nodes of $X_1$ and 
the strategic blue agents occupy the nodes of $X_2$. Then, by definition we have that 
$|\{ vz \in Y: \chi_v(\vv) \neq \chi_z(\vv)\}| = 
|\{ vz \in Y: v \in X_1, z \in X_2\}| \leq W$, and by~\eqref{eq:SW-equal-min-cut} we obtain
\begin{align*}
\SW(\vv) \geq \frac{2}{d_0} (|Y| - W ) = \xi.
\end{align*}

Conversely, assume that there exists an assignment $\vv$ with 
$\SW(\vv) \geq \xi = \frac{2}{d_0} ( |Y| - W )$.
Let $X_1$ consist of the nodes occupied by strategic red agents, and 
let $X_2$ consist of the nodes occupied by strategic blue agents. 
Then, $X_1 \cap X_2 = \varnothing$, and since there is an equal number 
of strategic red and blue agents, we also have that $|X_1| = |X_2|$. 
By~\eqref{eq:SW-equal-min-cut}, it is 
$$
|\{ vz \in Y: \chi_v(\vv) \neq \chi_z(\vv)\}| \le W,$$ 
and consequently
\begin{align*}
|\{ vz \in Y: \chi_v(\vv) \neq \chi_z(\vv)\}| = |\{ vz \in Y: v \in X_1, z \in X_2\}| \leq W,
\end{align*}
as desired.
\end{proof}


\section{Degree of Integration}
We now investigate whether equilibrium assignments can be diverse, by bounding the price of 
anarchy and stability in terms of the degree of integration; recall that this benchmark 
counts the number of agents who are exposed, i.e., have at least one neighbor of a different
type. As in the previous section, we again focus on games with strategic agents only.

We start by showing that the integration price of anarchy of $k$-swap games is $n/k$, i.e., 
it scales linearly with the number of agents. This indicates that 
in the worst case agents of different types 
are highly segregated, but, as the number of types increases,  
equilibria become more diverse and the price of anarchy decreases.

\begin{figure}[t]
\centering
\includegraphics[scale=0.4]{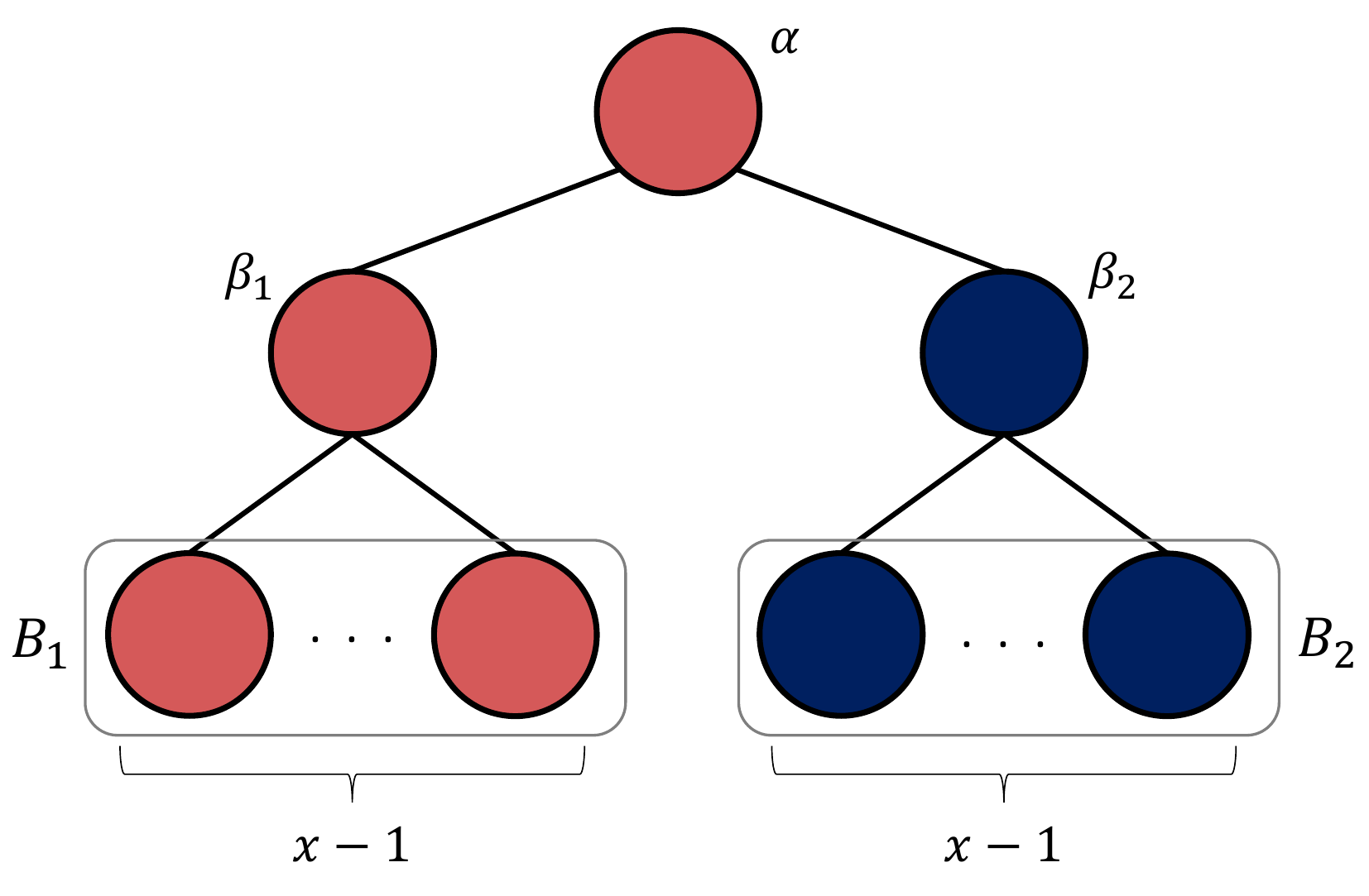}
\caption{The topology and the only possible equilibrium assignment 
for the $2$-swap game considered in the proofs of 
Theorems~\ref{thm:integration:poa} and~\ref{thm:integration:pos}. 
For $k$-swap games (Theorem~\ref{thm:integration:poa}) there
are $k$ identical subtrees, 
and in the worst equilibrium each subtree is filled by agents of a different type.}
\label{fig:integration}
\end{figure}

\begin{theorem}\label{thm:integration:poa}
For any $k \geq 2$, the integration price of anarchy of $k$-swap games with strategic 
agents only is at most $n/k$, and this bound is tight.
\end{theorem}
\begin{proof}
For the upper bound, consider a $k$-swap game with $n$ agents. 
By definition, the optimal degree of integration is at most $n$. 
Since the topology is a connected graph, in any assignment $\vv$ 
at least one agent per type must be exposed. 
Hence, $\DI(\vv) \geq k$, and the integration price of anarchy is at most $n/k$.

For the lower bound, consider a $k$-swap game with $n=kx+1$ agents such that there are 
$x+1$ agents of type $T_1$ and $x$ agents of type $T_\ell$ for every $\ell \in 
[k]\setminus\{1\}$. The topology is a tree with root node $\alpha$ that has $k$ children 
nodes $\beta_1, \dots, \beta_k$, each of which has $x-1$ children leaf nodes of its own; 
see Figure~\ref{fig:integration} for an example of this topology for $k=2$.

One can assign the agents to the nodes of the topology so that each agent 
is exposed; thus the maximum possible 
degree of integration is $n$. However, there is an equilibrium assignment $\vv$ in which 
$\alpha$ is occupied by an agent of type $T_1$ and for each $\ell\in [k]$
all nodes of the $\beta_\ell$-subtree 
are occupied by agents of type $T_\ell$. 
In $\vv$ only the agent in $\alpha$ and the agents in nodes $\beta_\ell$, 
$\ell \in [k]\setminus\{1\}$, are exposed, yielding degree of integration 
$\DI(\vv) = k$, and the bound follows.
\end{proof}

Next, we consider the integration price of stability. Using the same instance as 
in the proof of Theorem~\ref{thm:integration:poa}, we show a lower bound that 
depends linearly on the number of agents for the fundamental case of two agent types. This 
bound is tight by the previous theorem, and indicates that we cannot avoid ending up with 
equilibrium assignments in which the types are highly segregated, even in the best case.

\begin{theorem}\label{thm:integration:pos}
The integration price of stability of $2$-swap games with strategic agents only is at least 
$n/2$.
\end{theorem}

\begin{proof}
Consider a $2$-swap game with $x+1$ red agents and $x$ 
blue agents, for a total of $n=2x+1$ agents. 
The topology is the same as in Theorem 8: a tree 
consisting of a root node $\alpha$ with two children nodes $\beta_1$ and $\beta_2$, each of 
which has $x-1$ children leaf nodes of its own (sets $B_1$ and $B_2$); see 
Figure~\ref{fig:integration}. The optimal degree of integration is $n$.
We will now argue that the unique equilibrium assignment $\vv$ (up to symmetries) is such 
that $\alpha$ and all nodes of the $\beta_1$-subtree are occupied by red agents, while the 
nodes of the $\beta_2$-subtree are occupied by blue agents. 
The degree of integration of $\vv$ is exactly $2$, so the theorem follows.

Assume that agent $\pi_\alpha(\vv)$ is blue rather than red. We distinguish 
the following cases with regards to the agents occupying nodes $\beta_1$ and $\beta_2$.
\begin{itemize}
\item
Both $\pi_{\beta_1}(\vv)$ and $\pi_{\beta_2}(\vv)$ are of the same type. 
Assume that both of these agents are blue; as there are $x+1$ red agents, 
there must be red agents at the leaf nodes of both $B_1$ and $B_2$. 
But then agent $\pi_{\beta_2}(\vv)$ and some red agent occupying a node of $B_1$ 
can swap to increase their utility 
from strictly less than $1$ and zero to $1$ and positive, respectively. 
Hence, it must be the case that $\pi_{\beta_1}(\vv)$ and 
$\pi_{\beta_2}(\vv)$ are both red. Again, if there are blue agents
at the leaf nodes of both $B_1$ and $B_2$, the assignment is not an equilibrium, 
so the nodes of one of these subtrees (say, $B_1$) are all occupied by red agents, 
and the nodes of the other subtree are all occupied by blue agents. However, such an assignment 
cannot be an equilibrium since the blue agent $\pi_\alpha(\vv)$ 
and the red agent $\pi_{\beta_2}(\vv)$ both get zero utility and have an incentive to swap.
\item
$\pi_{\beta_1}(\vv)$ is red and $\pi_{\beta_2}(\vv)$ is blue. 
Since there are $x$ red agents remaining, at least one of them must be in $B_2$.
But then she can swap positions 
with the blue agent $\pi_\alpha(\vv)$ so that they increase their utility 
from zero and $1/2$ to $1/2$ and $1$, respectively. 
\end{itemize}
Therefore, agent $\pi_\alpha(\vv)$ must be red. 
Similarly to the previous case, we observe that 
if $\pi_{\beta_1}(\vv)$ and $\pi_{\beta_2}(\vv)$ are both blue, 
there must be red agents in both $B_1$ and $B_2$, and
if $\pi_{\beta_1}(\vv)$ and $\pi_{\beta_2}(\vv)$ are both red, 
there must be blue agents in both $B_1$ and $B_2$,
which means that $\pi_{\beta_1}(\vv)$ and some agent in 
a node of $B_2$ would have an incentive to swap.
Thus, one of $\pi_{\beta_1}(\vv)$ and $\pi_{\beta_2}(\vv)$ 
(say, $\pi_{\beta_1}(\vv)$) must be red and the other one must 
be blue. Then, if there is a blue agent in $B_1$, by a counting argument 
there is also a red agent in $B_2$, and these two agents would have an incentive to swap. 
Hence, $\vv$ is the only equilibrium assignment.
\end{proof}

To develop better intuition for the integration price of anarchy and stability, we also consider
the special case where the topology is a line. In this case, while the integration price 
of anarchy remains linear in $n/k$, the integration price of stability
can be bounded by a small constant.

\begin{theorem}\label{thm:integration:line}
Consider a $k$-swap game with strategic agents only, at least two agents per type, 
and a line topology.
The integration price of anarchy is at most $\frac{n}{2k-2}$, 
while the integration price of stability is at most $\frac94$. 
Moreover, if the number of agents of each type grows with $n$, 
the integration price of stability is at most $\frac32+o(1)$.
All these bounds are tight.
\end{theorem}

\begin{proof}
Let the topology be a line, with nodes $1, \dots, n$ connected in this order.

\medskip

\noindent
{\bf Price of anarchy.}
For the upper bound, consider an equilibrium assignment $\vv$. 
For each type $T_i$ let $\ell_i$ be the leftmost agent of type $T_i$ and let $r_i$ be the rightmost agent 
of this type. If $v_{\ell_i}\neq 1$, then $\ell_i$ has a neighbor to the left who
belongs to a different type; similarly, if $v_{r_i}\neq n$ then $r_i$ has a neighbor
to the right who belongs to a different type. Since nodes $1$ and $n$ are occupied
by exactly two agents, it follows that $\DI(\vv)\ge 2k-2$. 
As at most $n$ agents are exposed in any assignment, the bound follows.

To see that this bound is tight, it suffices to consider 
a $k$-swap game with $s$ agents per type, for some $s\ge 2$. We can create $s$
identical blocks of agents, with each block containing exactly one agent of each type, 
and place them on the line one after the other,  
so that every agent is exposed. 
However, there is also an equilibrium where agents are placed in monochromatic blocks 
of size $s$, so that only $2k-2$ agents are exposed.

\medskip

\noindent
{\bf Price of stability.}
We partition the agents of each type into blocks of size $2$
and $3$, with at most one block of size $3$ per type (that is, we create a block of size $3$
if and only if the number of agents of that type is odd); let $B_1, \dots, B_d$
be the resulting collection of blocks.
Observe that any assignment where agents
in each block are placed contiguously on the line is an equilibrium. 
Indeed, under any such assignment each agent has at least one neighbor of the same type, 
and no agent can move to a position where she would have two neighbors of her type;
she cannot move to nodes $1$ or $n$ either, since
agents at these nodes are unwilling to swap (they have utility $1$).

It remains to explain how to place these blocks on the line to maximize the degree of integration.
We do so greedily, from left to right. 
That is, we first pick some $i \in [k]$ such that $|T_i|\ge |T_j|$
for all $j\in [k]$, and place some block $B\subseteq T_i$ first;
if $|B|=2$, we assign the agents in $B$ to nodes $1$ and $2$, 
and if $|B|=3$, we assign the agents in $B$ to nodes $1$, $2$, and $3$.
Now, suppose that $\ell$ blocks have been placed, so that the last occupied node is node
$r$, and the agent there is of type $T_x$. 
For each $j\in [k]$, let $t_j$ be the number of agents of type $T_j$ 
who have not yet been placed. 
If $t_j=0$ for all $j\in [k]\setminus\{x\}$, we complete the assignment by simply 
placing all the remaining agents of type $T_x$ on the line.
Otherwise, we pick an
$i\in [k]\setminus\{x\}$ such that $t_i\ge t_j$ for all $j\in [k]\setminus\{x\}$, 
and place some block $B\subseteq T_i$ in positions $r+1, \dots, r+|B|$.

Let us say that a type $T_i$ is {\em dominant} if $|T_i|>n/2$.
An easy inductive argument shows that if no type is dominant, then 
under this assignment we never place two blocks of the same type next to each other;
the key observation is that if no type is dominant after $\ell$ blocks
have been placed, this remains to be the case after $\ell+2$ blocks have been placed,
and hence if we still have at least two blocks to place, 
we have at least two types to choose from.  
In this case, the only agents who are not exposed are agents at nodes 
$1$ and $n$ as well as agents located at the middle of a block of size $3$, i.e., at most $k+2$ 
agents. Thus, the integration price of stability is at most $\frac{n}{n-k-2}$ in this case. 
Now, suppose that some type (say, type $T_1$) is dominant. 
If there are $\lambda$ blocks of types $T_2, \dots, T_k$, 
then under our procedure we will first alternate between blocks of type $T_1$ and blocks
of other types, and then place the remaining blocks of type $T_1$ (if any).
Then at least $4\lambda$ agents will be exposed.
On the other hand, at most $k-1$ of these $\lambda$ blocks are of size $3$, so
we have at most $2\lambda+k-1$ agents of types $T_2, \dots, T_k$, and hence
in any assignment at most $2(2\lambda+k-1)$ agents of type $T_1$ can be exposed.
Thus, the integration price of stability in this case is at most 
$$\frac{3(2\lambda+k-1)}{4\lambda}=\frac{3}{2}+\frac{3k-3}{4\lambda}.$$ 
Since $\lambda \ge k-1$, 
this quantity is at most $\frac94$. Further, 
if we assume that the number of agents of each type grows with $n$, we have
$\frac{3k-3}{4\lambda}=o(1)$, and the bound becomes $\frac32+o(1)$.

To see that the bound $\frac94$ on the integration price of stability is tight, 
consider a game with six red agents
and three blue agents. In equilibrium, the three blue
agents need to form a single block: if there is an isolated blue agent $b$,
there is also a red agent $r$ who is not a neighbor of $b$, but has another blue
neighbor; $b$ and $r$ can then benefit from swapping. 
Thus, in equilibrium at most two blue agents---and hence at most two red agents---are exposed.
However, we can also create three blocks of agents, 
with each block consisting of a single blue agent surrounded by two red agents, 
and place these three blocks consecutively on the line, so that each agent 
is exposed.

To see that the bound $\frac32$ is tight if the number of agents of each type
grows with $n$, consider an instance with $4s$ red agents and $2s$ blue agents, 
for some $s\in{\mathbb N}$. Arguing as above, we can see that in equilibrium the blue
agents have to appear in blocks of size at least $2$, so that each blue agent
has at most one red neighbor. Hence, at most $2s$ red agents have a blue neighbor, 
and thus the number of exposed agents is at most $4s$. 
On the other hand, 
by placing agents in red-blue-red blocks, as described in the previous paragraph,
we can ensure that all $6s$ agents are exposed.  
\end{proof}

\noindent
Hence, for games with simple line topologies, integration can be achieved in equilibrium. 
However, when left to their own devices, the agents may end up in a very segregated configuration.

We conclude this section by studying the complexity of computing assignments (not necessarily 
equilibria) with high degree of integration. Unfortunately, it turns out that even this 
task is computationally intractable.

\begin{theorem}
Given a $k$-swap game, computing an assignment in which every agent 
is exposed is NP-complete, 
even if $k=2$ and all agents are strategic. 
\end{theorem}

\begin{proof}
The problem is clearly in NP; for a given assignment we can verify whether 
each of the $n$ agents has at least one neighbor of a different type in time $O(n^2)$. 
For the NP-hardness proof, we give a reduction from the {\sc Vertex Cover} problem, 
which is known to be NP-complete.
An instance of {\sc Vertex Cover} consists of an undirected graph $H = (X, Y)$ 
and an integer $\lambda \le |X|$. 
It is a yes-instance if there exists a set $X' \subseteq X$ such that $|X'| \le \lambda$ 
and $\{v,w\} \cap X' \neq \varnothing$ for every edge $\{v,w\} \in Y$. 
Such a set $X'$ is called a vertex cover of $H$.
Without loss of generality, we assume that $H$ has no isolated vertices.

Given an instance $\langle H, \lambda \rangle$ of the {\sc Vertex Cover} problem with 
$H=(X,Y)$, we construct a $2$-swap game as follows:
\begin{itemize}
\item
We have $|X|+|Y|-\lambda$ red agents and $\lambda$ blue agents,  
for a total of $n=|X|+|Y|$ agents.   
\item 
To construct the topology $G=(V,E)$, we start with the graph $H$. 
Then, for each edge $e = \{v,w\}\in Y$, we add a node $s_e$, 
and two edges connecting $s_e$ to $v$ and $w$. Let $Q = \{s_e: e \in Y\}$.
Then, $V = X \cup Q$, $X \cap Q = \varnothing$, and $|V| = |X|+|Q| = |X|+|Y| = n$.
\end{itemize}
We show that $H$ has a vertex cover of size at most $\lambda$ if and only if there exists 
an assignment in which every agent is exposed.

First, suppose that there exists a vertex cover $X'\subseteq X$ of $H$ of size at most $\lambda$;
by adding nodes, we can assume that $|X'|=\lambda$. 
Consider the assignment $\vv$ in which the nodes of $X'$ 
are occupied by blue agents, while all other nodes of $V \setminus X'$ 
are occupied by red agents. In this assignment, every agent is exposed:
\begin{itemize}
\item
For every blue agent $i$ occupying a node $v_i \in X'$, since $H$ has no isolated nodes, 
there must exist an edge $e \in Y$ such that $v_i \in e$, and hence $v_i$ is connected to 
$s_e$, which is occupied by a red agent.
\item
For every red agent $i$ occupying a node $v_i \in X \setminus X'$, since $X'$ is a vertex 
cover, $v_i$ must be connected to a node $z \in X'$, which is occupied by a blue agent.
\item 
For every red agent $i$ occupying a node $v_i=s_e\in Q$, since $X'$ is a vertex cover, at 
least one of $e$'s endpoints must be in $X'$, which is occupied by a blue agent, and  
$s_e$ is connected to both endpoints of $e$.
\end{itemize}

Conversely, suppose that $\vv$ is an assignment of the agents to the nodes of the topology 
such that every agent is exposed. 

For each edge $e = \{v, w\} \in E$, let $\ell(e)$ be an arbitrary element of $\{v, w\}$.
Let $X' = \{ v\in X: \pi_v(\vv) \text{ is blue}\}$,
$X_Q = \{z\in X: z=\ell(e)\text{ for some $e$ such that $\pi_{s_e}(\vv)$ is blue}\}$.
Since there are $\lambda-|X'|$ nodes in $Q$ that are occupied by blue agents, we have
$|X_Q|\le \lambda-|X'|$ and hence $|X'\cup X_Q|\le \lambda$.
We claim that $X'\cup X_Q$ is a vertex cover for $H$.
Indeed, consider an arbitrary edge $e = \{v, w\}\in E$; we will argue that 
$e\cap(X'\cup X_Q)\neq\varnothing$. If $v\in X'$ or $w\in X'$, 
we are done. Otherwise, both $v$ and $w$ are occupied by red agents;
since $\pi_{s_e}(\vv)$ is adjacent to an agent
of a different type, it follows that $s_e$ is occupied by a blue agent
and $\ell(e)\in X_Q$. Hence one of $v$ and $w$ is in $X_Q$.
This completes the proof.
\end{proof}

\section{Conclusions and Open Problems}
We have studied Schelling games on graphs in which pairs of agents 
can deviate by swapping their locations. We considered questions 
related to the existence and the efficiency 
of equilibrium assignments, both from a social welfare perspective and from
a diversity perspective.

While equilibria are known to exist in instances where the topology is highly 
structured, we showed that their existence is not guaranteed in general, 
and deciding whether a given swap game admits an equilibrium assignment
is NP-complete. Even though we have implicitly assumed that the tolerance 
threshold of every agent is $1$, and thus she is never truly happy unless she is 
connected to friends only, our proofs extend to other threshold values 
as well. For instance, one can verify that Theorem~\ref{thm:non-existence} for 
$k=2$ holds for any $\tau \in (2/3,1)$. A challenging open 
question is to completely characterize the topologies and the threshold values for which 
equilibria are guaranteed to exist, and also design efficient algorithms to compute equilibria
when they exist.

We have introduced a new index for measuring the diversity of a given assignment,
which we called the degree of integration. It would be interesting to explore
the tradeoffs between diversity and social welfare: can we compute (equilibrium) assignments
with a given degree of integration that maximize the social welfare? While our results
indicate that this problem is hard for general topologies, one could hope to obtain
approximate or parameterized algorithms, or focus on simple topologies. One can also 
investigate more ambitious diversity indices, e.g., by considering, for each agent, 
the number of other types she is exposed to. 

\bibliographystyle{named}
\bibliography{schelling-bib}

\end{document}